\documentclass[letterpaper, 10 pt, conference]{ieeeconf}  

\IEEEoverridecommandlockouts                              
\usepackage{amsmath}
\usepackage{amssymb} 
\usepackage{algpseudocode}
\usepackage{caption}
\usepackage{subcaption}
\usepackage{enumerate}

\newcommand{\BEAS}{\begin{eqnarray*}}
\newcommand{\EEAS}{\end{eqnarray*}}
\newcommand{\BEA}{\begin{eqnarray}}
\newcommand{\EEA}{\end{eqnarray}}
\newcommand{\BEQ}{\begin{equation}}
\newcommand{\EEQ}{\end{equation}}
\newcommand{\BIT}{\begin{itemize}}
\newcommand{\EIT}{\end{itemize}}
\newcommand{\BNUM}{\begin{enumerate}}
\newcommand{\ENUM}{\end{enumerate}}

\newcommand{\ie}{{\it i.e.}}

\newcommand{\cl}{\mathop{\bf cl}}

\newcommand{\ip}[2]{\langle {#1},{#2} \rangle}

\newcommand{\trace}[1]{\mathop{\bf Tr}\left({#1}\right)}
\newcommand{\rank}[1]{\mathop{\bf Rank}\left({#1}\right)}

\usepackage[amsmath,thmmarks]{ntheorem}
\RequirePackage{amssymb}

\theoremstyle{plain}
\theoremheaderfont{\normalfont\bfseries}
\theorembodyfont{\itshape}
\theoremseparator{:}
\theoremsymbol{}
\newtheorem{theorem}{Theorem}
\newtheorem{proposition}{Proposition}
\newtheorem{lemma}{Lemma}
\newtheorem{corollary}{Corollary}
\newtheorem{example}{Example}

\theoremstyle{nonumberplain}
\theoremheaderfont{\scshape}
\theorembodyfont{\normalfont}
\theoremsymbol{\ensuremath{_\blacksquare}}
\qedsymbol{\ensuremath{_\blacksquare}}

\title{\LARGE \bf
Primal robustness and semidefinite cones
}

\author{Seungil You, Ather Gattami, and John C. Doyle
\thanks{This work was supported by Kwanjeong Graduate Fellowship}
\thanks{S. You and J. Doyle are with the Control and Dynamical Systems, California Institute of Technology, Pasadena, CA 91125, USA
		{\tt\small \{syou, doyle\}@caltech.edu}
		A. Gattami is with Ericsson Research, Stockholm, Sweden, {\tt\small ather.gattami@ericsson.com}}
}

\begin{document}

\maketitle
\thispagestyle{empty}
\pagestyle{empty}

\begin{abstract}
This paper reformulates and streamlines the core tools of robust stability and performance for LTI systems using now-standard methods in convex optimization.  In particular, robustness analysis can be formulated directly as a primal convex (semidefinite program or SDP) optimization problem using sets of gramians whose closure is a semidefinite cone.  This allows various constraints such as structured uncertainty to be included directly, and worst-case disturbances and perturbations constructed directly from the primal variables.   Well known results such as the KYP lemma and various scaled small gain tests can also be obtained directly through standard SDP duality.  To readers familiar with robustness and SDPs, the framework should appear obvious, if only in retrospect.  But this is also part of its appeal and should enhance pedagogy, and we hope suggest new research.  There is a key lemma proving closure of a grammian that is also ÔÕobviousÕÕ but our current proof appears unnecessarily cumbersome, and a final aim of this paper is to enlist the help of experts in robust control and convex optimization in finding simpler alternatives.
\end{abstract}

\section{Introduction}
Robust control theory has been an important subject in the control engineering both in theory and practice \cite{Dullerud:2010tc}.
Theoretical developments has been evolved in various flavor, 
but the most modern approach to this subject is the Linear Matrix Inequality (LMI) based approach \cite{boyd1995linear}.

In order to obtain the LMI characterization of system behavior, so called the $\mathcal{S}$-procedure \cite{yakubovich1971s} has been extensively used.
Motivated by the popularity of LMIs in systems theory, semidefinite programming (SDP) duality has been used to understand such LMIs and control theoretic interpretation of the SDP duals of these LMIS has been reported \cite{balakrishnan2003semidefinite, scherer2006lmi, ebihara2009robust}.
In particular, the dual LMI approach is used to extract the worst case frequency variable and disturbance in \cite{scherer2006lmi, ebihara2009robust}.
However, a recent paper shows that this dual LMI has its own right as a well-defined optimization when it comes to $\mathcal{H}_{\infty}$ analysis \cite{gattami2013simple}.
In  \cite{you2014h}, we also report that the Kalman--Yakubovich--Popov (KYP) lemma \cite{Rantzer:2011wn} is an SDP dual of this optimization problem, which should be obvious to the experts.
An interesting observation in here is that the dual LMI may be another starting point of robustness analysis which does not require $\mathcal{S}$-procedure,
and well known results, such as the KYP lemma can be obtained through SDP duality, \ie, reversing the theoretical developments.

To this end, this paper provides a complete characterization of gramians generated by a linear time invariant (LTI) system.
It turns out the closure of a set of gramians is an intersection of a subspace and a semidefinite cone.
The seminal paper \cite{georgiou2002structure} attempts to obtain similar results on the covariance matrices generated by stochastic disturbances,
but this paper characterizes gramians from {\it{deterministic}} disturbances, which is suitable for existing robustness results.
More importantly, we provide a semidefinite representation of gramians in contrast to the rank constraint in \cite{georgiou2002structure}.
This semidefinite representation allows us to formulate extended $\mathcal{H}_{\infty}$ analysis, where we can directly capture numerous prior information on the disturbance including structural properties, as an SDP.

In addition, the SDP dual of our primal optimization provides the well known LMI characterization of the system behavior.
We exemplify this procedure for the KYP lemma, but the result can be easily extended to more general disturbance setting, and our approach provides an arguably simple proof through standard SDP duality.
In addition, SDP duality also provides the scaled small gain tests for the robust stability verification.
However, our primal formulation provides a specific input-output pair proving that the system is not robustly stable, which is not a trivial task in the scaled small gain test.
This is because the variables in the scaled small gain test does not contain useful input-output information, so if the test fails, it is hard to extract a specific pair that disproves robust stability.
This entire procedure of obtaining LMIs for robustness analysis should be obvious to the experts, which shows a pedagogical benefit of our approach.
We also hope that our new tool opens up a new research direction in robust control theory. 

\subsection{Notation}
$\mathbb{H}^n, \mathbb{H}_+^n, \mathbb{H}_{++}^n$ are sets of $n \times n$ Hermitian, positive semidefinite, positive definite matrices, respectively.
The generalized inequality $X \succeq 0$ means $X \in \mathbb{H}_+$, and $X \succ 0$ means $X \in \mathbb{H}_{++}$.
We use $l_2$ for $l_2[0,\infty)$, the Hilbert space of square summable sequence with the starting index $0$.
The bold Latin letter $\mathbf{x}$ means a sequence in $l_2$.
In addition, for a vector and vector-valued signal, $\|\cdot\|_2$ is the two norm.
For a matrix and linear operator, $\|\cdot\|_F$ is the Frobenius norm,  and
$\rho(A)$ denotes the spectral radius of $A$, $A^*$ denotes a Hermitian/Adjoint operator.

\section{A semidefinite representation of Gramians}
For a signal $\mathbf{u} \in l_2$, we define the gramian $\Lambda : l_2 \rightarrow \mathbb{H}_+$, 
\BEAS
\Lambda(\mathbf{u}) = \sum_{k=0}^{\infty}u_k u_k^*.
\EEAS
Notice that the gramian is well-defined because each entry of the matrix is finite, and $\mathbb{H}_+$ is closed.

For notational convenience, let $\Lambda(\mathbf{u},\mathbf{v}) =  \Lambda\left(\begin{bmatrix}\mathbf{u}\\ \mathbf{v} \end{bmatrix}\right) = \sum_{k=0}^{\infty}\begin{bmatrix} u_k\\v_k\end{bmatrix}\begin{bmatrix} u_k\\v_k\end{bmatrix}^*$, and $\Lambda_N(\mathbf{u})$  be a finite truncation of $\Lambda$, $\Lambda_N(\mathbf{u}) = \sum_{k=0}^{N-1} u_ku_k^*$. 

Suppose we have matrices $A \in \mathbb{C}^{n\times n}$, $B \in \mathbb{C}^{n \times m}$, where $A$ is Schur stable, $\rho(A) < 1$, and
let $\mathbf{M} : l_2^m \rightarrow l_2^n$ be a linear operator such that $\mathbf{x} = \mathbf{M}\mathbf{w}$ if
\begin{align}
x_{k+1} &= Ax_{k} + B w_k \label{eq:state}\\
x_0 &= 0. \label{eq:initial}
\end{align}

In this paper, we consider a set of gramians, $\mathcal{S}$, generated by $\mathbf{M}$,
\BEAS
\mathcal{S} := \left\{V \in \mathbb{H}_+ ~: V = \Lambda(\mathbf{Mw}, \mathbf{w}) \text{ for some $\mathbf{w} \in l_2$}\right\}.
\EEAS
In terms of $\mathbf{x}, \mathbf{w}$, the gramian can be seen as
\BEAS
V =  \begin{bmatrix} X & R\\R^* & W \end{bmatrix} = \begin{bmatrix} \sum_k x_kx_k^* & \sum_k x_kw_k^*\\ \sum_k w_kx_k^* & \sum_kw_kw_k^* \end{bmatrix}.
\EEAS
From this definition, we can easily see that the gramian captures various input-output relationships in the system.
For example, 
 $\|\mathbf{w}\|^2_2 = \trace{W}$, and $\|C\mathbf{x} + D\mathbf{w}\|_2^2 = \trace{\begin{bmatrix} C^*C & C^*D\\D^*C & D^*D \end{bmatrix} V}$.

However, for a given matrix $V$, checking $V \in \mathcal{S}$ is not a trivial task, because one should search over the $l_2$ space, an infinite dimensional space, to find a signal $\mathbf{w}$ generating $V$.
Therefore it is desirable to find a convinient way to characterize the set $\mathcal{S}$.
Let us consider the following set
\BEAS
\mathcal{C} := \left\{V \in \mathbb{H}_+ : \begin{bmatrix} A & B \end{bmatrix} V \begin{bmatrix} A^* \\ B^* \end{bmatrix} =  \begin{bmatrix} I & 0 \end{bmatrix} V \begin{bmatrix} I \\ 0 \end{bmatrix}\right\}.
\EEAS
Notice that $\mathcal{C}$ is an intersection of a subspace in $\mathbb{H}$ and a semidefinite cone, therefore $\mathcal{C}$ is a finite dimensional closed, convex cone that is semidefinite programming (SDP) representable.
This means that checking $V \in \mathcal{C}$ can be easily done by a semidefinite programming.

Why do we need this set $\mathcal{C}$?
Is there any relationship between $\mathcal{C}$ and $\mathcal{S}$?
The following proposition shows an interesting observation between these two sets.
\begin{proposition}
The set $\mathcal{S}$ is a subset of $\mathcal{C}$.
\label{prop:easy_direction}
\end{proposition}
\begin{proof}
Suppose $V \in \mathcal{S}$.
This means there exists a signal $\mathbf{w} \in l_2$ such that $V = \Lambda(\mathbf{Mw}, \mathbf{w})$.
Let $\mathbf{x} = \mathbf{Mw}$. 
Then, since $x_{k+1} = Ax_k + Bw_k$, we have
\BEAS
x_{k+1}x_{k+1}^* &=& \begin{bmatrix} A & B \end{bmatrix}  \begin{bmatrix}x_k\\w_k\end{bmatrix}\begin{bmatrix}x_k\\w_k\end{bmatrix}^* \begin{bmatrix} A^* \\ B^* \end{bmatrix}
\EEAS
By taking an infinite sum, we have
\BEAS
\sum_{k=0}^{\infty} x_{k+1}x_{k+1}^* &=& \sum_{k=0}^{\infty} \begin{bmatrix} A & B \end{bmatrix}  \begin{bmatrix}x_k\\w_k\end{bmatrix}\begin{bmatrix}x_k\\w_k\end{bmatrix}^* \begin{bmatrix} A^* \\ B^* \end{bmatrix}\\
&=&  \begin{bmatrix} A & B \end{bmatrix} V\begin{bmatrix} A^* \\ B^* \end{bmatrix}.
\EEAS
Moreover, since $x_0 = 0$, $\sum_{k=0}^{\infty} x_{k+1}x_{k+1}^* = \sum_{k=0}^{\infty} x_kx_k^*$.
Using this fact with $\begin{bmatrix} I _n& 0_{n,m} \end{bmatrix} V \begin{bmatrix} I _n& 0_{n,m} \end{bmatrix}^* = \sum_{k=0}^{\infty}x_kx_k^*$, we can conclude that 
$V \in \mathcal{C}$.
\end{proof}

The above observation is somewhat trivial due to the convergence.
But the immediate, and important question arises. Is $\mathcal{C}$ equal to $\mathcal{S}$?
If this is true, we can replace a complicated set $\mathcal{S}$ by a semidefinite representable convex cone $\mathcal{C}$.
For an optimization with $\mathcal{S}$, this has a dramatic impact: any optimization involving $\mathcal{S}$ becomes a convex program.
Unfortunately, this is not the case.
\begin{example}
Let $A = a$, $B = 1$, where $|a| < 1$, and $a \neq 0$.
Consider $V = \begin{bmatrix} \frac{1}{1-a}\\ 1\end{bmatrix}\begin{bmatrix} \frac{1}{1-a}\\ 1\end{bmatrix}^*$.
It can be easily checked $V \in \mathcal{C}$, and $\rank{V} = 1$, but $V \not\in \mathcal{S}$.
%
%
\end{example}

However, we have a surprising fact: $\mathcal{C}$ is the {\it{closure}} of $\mathcal{S}$.
Although the idea of the proof is simple, but our current proof goes through a tedious analysis to apply the $\epsilon-\delta$ style argument.
So we present the main result here, and relegate the sketch of the proof to the appendix.
But we want to emphasize that our proof is {\it{constructive}}, therefore given $V$, we can find a signal $\mathbf{w}$ that approximates $V$ arbitrarily close.

%

\begin{lemma}
Suppose $V = \begin{bmatrix} X & R\\R^* & W \end{bmatrix} \in \mathcal{C}$.
Then for all $\varepsilon > 0$, there exists $\mathbf{w} \in l_2$ with finite number of non-zero entries such that
\begin{align}
&\|\Lambda(\mathbf{Mw},\mathbf{w}) - V\|_F < \varepsilon \label{eq:small-error-full}\\
&\Lambda(\mathbf{w}) = W  \label{eq:exactw-full}
\end{align}
\vspace{-0.5cm}
\label{prop:approximation}
\end{lemma}

The following is immediate.

\begin{lemma}
$\mathcal{C} = \cl{\mathcal{S}}$.
\label{prop:closure}
\end{lemma}
\begin{proof}
From Proposition \ref{prop:easy_direction}, $\mathcal{S} \subset \mathcal{C}$, which shows $\cl{\mathcal{S}} \subset \mathcal{C}$ since $\mathcal{C}$ is closed.
From Lemma \ref{prop:approximation}, $\mathcal{C} \subset \cl{\mathcal{S}}$.
Therefore, $\mathcal{C} = \cl{\mathcal{S}}$.
\end{proof}

The above two results are key lemmas in this paper.
Many robustness analysis  can be stated as an optimization problem with gramians, that has a linear objective function.
Therefore, a set of gramians, $\mathcal{S}$, can be replaced by $\mathcal{C}$ without any conservatism, and more importantly the resulting problem becomes an SDP.

%


Before concluding this section, we present a connection between the controllability of $(A,B)$ and the geometric property of the SDP cone $\mathcal{C}$.

\begin{proposition}
There exists  $V \in \mathcal{C} \cap \mathbb{H}_{++}$ if and only if $(A,B)$ is controllable.
\label{prop:nonempty}
\end{proposition}
\begin{proof}
Suppose $(A,B)$ is controllable.
Since $A$ is stable, the controllability gramian $W_c$
\BEAS
AW_cA^* - W_c + BB^* = 0
\EEAS
is positive definite.
Let
\BEAS
V = \begin{bmatrix} W_c & 0 \\ 0 & I \end{bmatrix},
\EEAS
then $V \in \mathcal{C} \cap \mathbb{H}_{++}$.

Now suppose there exists $V \in \mathcal{C} \cap \mathbb{H}_+$.
Since $V \in \mathcal{C}$,
\BEAS
AXA^* - X + BR^*A^* + ARB^* + BWB^* = 0.
\EEAS
Let $W = TT^*$, and $\tilde{B} = BT$, $K = T^{-1}R^*X^{-1}$.
Then,
\BEAS
(A+\tilde{B}K)X(A+\tilde{B}K)^* - X + \tilde{B}\tilde{B}^* = 0.
\EEAS
Since $X \succ 0$, $(A+\tilde{B}K, \tilde{B})$ is controllable, and it is easy to check that $(A,B)$ is controllable.
\end{proof}

Now we use all these results to prove well-known results in robust control theory which shows the effectiveness of our new, primal approach.

\section{$\mathcal{H}_{\infty}$ analysis and the KYP lemma}
\subsection{$\mathcal{H}_{\infty}$ analysis}
In $\mathcal{H}_{\infty}$ analysis, we would like to find  the worst-case disturbance that maximizes the output norm.
Specifically, let $z_{k} = Cx_{k} + Dw_{k}$. Then we want to solve
\begin{equation}
\begin{aligned}
\mu_{\infty} := & \underset{\mathbf{w},\mathbf{x},\mathbf{z}} {\text{maximize}}
& & \|\mathbf{z}\|_2^2\\
&\text{subject to}
&& x_{k+1} = Ax_k + Bw_k, \qquad x_0 = 0\\
&&& \|\mathbf{w}\|_2^2 = 1.
\end{aligned}
\label{eq:original-h-infinity}
\end{equation}
The optimal value of \eqref{eq:original-h-infinity} is the square of the $\mathcal{H}_{\infty}$ norm of the system.

Let $ \Lambda(\mathbf{Mw}, \mathbf{w}) = \begin{bmatrix} X & R\\R^* & W\end{bmatrix}$, where $X \in \mathbb{H}^n$, $W \in \mathbb{H}^m$.
Notice that
\BEAS
\|\mathbf{z}\|^2_2 &=&\trace{\begin{bmatrix}C^*C & C^*D\\D^*C&D^*D\end{bmatrix}\begin{bmatrix} X & R\\R^* & W\end{bmatrix}}\\
\|\mathbf{w}\|^2_2 &=& \trace{W}.
\EEAS

This shows that the optimization \eqref{eq:original-h-infinity} is equivalent to
\begin{equation}
\begin{aligned}
& \underset{X,R,W} {\text{maximize}}
& & \trace{\begin{bmatrix}C^*C & C^*D\\D^*C&D^*D\end{bmatrix}\begin{bmatrix} X & R\\R^* & W\end{bmatrix}}\\
&\text{subject to}
&& \begin{bmatrix} X & R\\R^* & W\end{bmatrix} \in \mathcal{S}\\
&&& \trace{W} = 1.
\end{aligned}
\label{eq:nonconvex-h-infinity}
\end{equation}

However the above problem is hard to solve because the feasible set $\mathcal{S}$ is involved with an infinite dimensional, $l_2$ space.
Using Lemma \ref{prop:approximation} and \ref{prop:closure}, we can replace $\mathcal{S}$ by $\mathcal{C}$ which results in an SDP that computes the  $\mathcal{H}_{\infty}$ norm of the system.

\begin{proposition}
Define the set $\mathcal{F} = \{V=\begin{bmatrix} X & R\\R^* & W\end{bmatrix} \in \mathbb{H}:\trace{W} = 1\}$.
Then $\cl({\mathcal{S} \cap \mathcal{F}}) = \mathcal{C} \cap \mathcal{F}$.
\label{prop:h-infinity}
\end{proposition}
\begin{proof}
Since $\cl{\mathcal{S}} = \mathcal{C}$, and $\cl{\mathcal{F}} = \mathcal{F}$, we have $\cl{(\mathcal{S} \cap \mathcal{F})} \subset \mathcal{C} \cap \mathcal{F}$.
Now suppose $V \in \mathcal{C} \cap \mathcal{F}$.
From Lemma \ref{prop:approximation}, for any $\epsilon > 0$, there exists $\mathbf{w} \in l_2$ such that
$\Lambda(\mathbf{w}) = 1$, and $\|V-\Lambda(\mathbf{Mw}, \mathbf{w})\|_F < \epsilon$.
Since $\Lambda(\mathbf{Mw}, \mathbf{w}) \in \mathcal{S} \cap \mathcal{F}$, we can conclude that $V \in \cl({\mathcal{S} \cap \mathcal{F}})$.
\end{proof}

Therefore, from the continuity of $\trace{\cdot}$, we can replace $\mathcal{S}$ in \eqref{eq:nonconvex-h-infinity} by $\mathcal{C}$ without any conservatism, and this is the main reason why we can compute the $\mathcal{H}_{\infty}$ norm using an SDP.

\begin{equation}
\begin{aligned}
& \underset{X,R,W} {\text{maximize}}
& & \trace{\begin{bmatrix}C^*C & C^*D\\D^*C&D^*D\end{bmatrix}\begin{bmatrix} X & R\\R^* & W\end{bmatrix}}\\
&\text{subject to}
&& X =  \begin{bmatrix} A & B \end{bmatrix}\begin{bmatrix} X & R\\R^* & W\end{bmatrix}\begin{bmatrix} A^*\\ B^* \end{bmatrix}\\
&&& \begin{bmatrix} X & R\\R^* & W\end{bmatrix}\succeq0, \trace{W} = 1.
\end{aligned}
\label{eq:convex-h-infinity}
\end{equation}
Clearly, the above optimization is an SDP that can be solved via a standard SDP solver \cite{toh1999sdpt3}, and has the same form as \cite{gattami2013simple, you2014h}.
More importantly, without going through an additional proof, we can easily check the equivalence between the finite-dimensional SDP \eqref{eq:convex-h-infinity} and an infinite-dimensional optimization \eqref{eq:nonconvex-h-infinity}.
Notice that after obtaining the optimal solution of \eqref{eq:convex-h-infinity}, we can construct a signal $\mathbf{w}$ that asymptotically achieves the optimal value using the proof of Lemma \ref{prop:approximation}.
This direct formulation approach will be repeated through this paper, and  our Lemma \ref{prop:approximation} provides a elementary, yet elegant approach to the classical problems in robust control theory.

Using SDP duality, we can expand our understanding of \eqref{eq:convex-h-infinity}.
The following is the SDP dual of \eqref{eq:convex-h-infinity}:
\begin{align}
& \underset{\lambda \in \mathbb{R}, P \in \mathbb{H}} {\text{minimize}}
& & \lambda \label{eq:convex-h-infinity-dual}\\
&\text{subject to}
&& \begin{bmatrix} A & B\\C & D\end{bmatrix}^*\begin{bmatrix} P & 0\\0 & I \end{bmatrix}\begin{bmatrix} A & B\\C & D\end{bmatrix}
- \lambda \begin{bmatrix} 0 & 0\\ 0 & I_m \end{bmatrix} \preceq 0. \nonumber
\end{align}
which is the same problem derived from the KYP lemma.
Since we provide a well-defined optimization using our set $\mathcal{C}$, we can claim that the KYP lemma is the dual of $\mathcal{H}_{\infty}$ analysis.
In addition, we can easily check that strong duality holds between \eqref{eq:convex-h-infinity} and \eqref{eq:convex-h-infinity-dual} because the dual program \eqref{eq:convex-h-infinity-dual} is strictly feasible.
%
\begin{proposition}
\eqref{eq:convex-h-infinity-dual} is strictly feasible.
\label{prop:strictdual}
\end{proposition}
\begin{proof}
Since $A$ is stable, there exists $P$ such that $A^*PA-P+C^*C \prec 0$.
Then by taking $\lambda$ large enough, we can find a strictly feasible point of \eqref{eq:convex-h-infinity-dual}.
\end{proof}

As a result, we have the following corollary from the Conic duality theorem \cite{ben1987lectures}.
\begin{corollary}
The duality gap between \eqref{eq:convex-h-infinity}, \eqref{eq:convex-h-infinity-dual} is zero, and the primal problem \eqref{eq:convex-h-infinity} is solvable.
\label{cor:primal-solvable}
\end{corollary}

However the dual optimum may not be attained.
Let us see the following example.
\begin{example}
Let $A = \frac{1}{2}, B = 0, C = 1, D=  1$.
Then the optimal solution of \eqref{eq:convex-h-infinity} is given by $V^{\star} = \begin{bmatrix} 0 & 0\\0 & 1 \end{bmatrix}$, and the corresponding optimal value is $+1$.
The optimal value of the dual \eqref{eq:convex-h-infinity-dual} is also $+1$, by taking $\lambda^{\star} = 1$, and $P^{\star} \rightarrow \infty$.
Clearly, the dual optimum is not attained.
\label{example:weird-pair}
\end{example}

The pair $(A,B)$ in the above example is not controllable, and this phenomena is closely related to the controllability assumption in the KYP lemma.
In order to ensure the existence of a multiplier $P$ (a dual optimal solution), we need the controllability assumption.
\begin{proposition}
The primal program \eqref{eq:convex-h-infinity} is strictly feasible if and only if $(A,B)$ is controllable.
\end{proposition}
\begin{proof} 
From Proposition \ref{prop:nonempty}, there exists $V \in \mathcal{C} \in \mathbb{H}_+$ if and only if $(A,B)$ is controllable.
\end{proof}

As a corollary, we have the following result on strong duality.

\begin{corollary}
Suppose $(A,B)$ is controllable. Then both the primal problem \eqref{eq:convex-h-infinity}, and dual \eqref{eq:convex-h-infinity-dual} are solvable.
\label{cor:solvable}
\end{corollary}

\subsection{A proof of KYP lemma}
Based on the observation between the KYP lemma and our primal optimization, it is very easy to prove the KYP lemma using the theorem of alternatives for SDP \cite{balakrishnan2003semidefinite}.

Let $\mathcal{A}$, $\mathcal{B}$ linear operators on $\mathbb{H}$.
Then,

\begin{theorem}[ALT4]
Exactly one of the following is true.
\begin{enumerate}[(i)]
\item There exists an $X \in \mathbb{H}$ with $\mathcal{A}(X) + A_0 \succ 0$, and $\mathcal{B}(X) = 0$.
\item There exists a non zero $Z \in \mathbb{H}_+$, $W \in \mathbb{H}$, $\mathcal{A}^{*}(Z) + \mathcal{B}^*(W) = 0$, and $\trace{A_0 Z} \leq 0$.
\end{enumerate}
\label{thm:alt4}
\end{theorem}
%
Now we are ready to prove the following theorem.
\begin{theorem}[KYP lemma, strict inequality]
The optimal value of \eqref{eq:original-h-infinity}, $\mu_{\infty} < 1$ if and only if there exists $P \in \mathbb{H}^n$ such that 
\begin{equation}
\begin{bmatrix} A^*PA-P & A^*PB\\B^*PA & B^*PB \end{bmatrix}+ \begin{bmatrix} C^*C & C^*D\\ D^*C & D^*D-I_m \end{bmatrix} \prec 0 \label{eq:kypp}.
\end{equation}
\label{thm:strict-kyp}
\end{theorem}
\begin{proof}
From the Theorem \ref{thm:alt4}, there exists $P$ with \eqref{eq:kypp} holds if and only if there is no non-zero $V \succeq 0$ such that
$\begin{bmatrix} A & B \end{bmatrix}V\begin{bmatrix} A & B \end{bmatrix}^*  = \begin{bmatrix} I & 0 \end{bmatrix}V\begin{bmatrix} I & 0 \end{bmatrix}^*$, and $
\trace{ \begin{bmatrix} C^*C & C^*D\\ D^*C & D^*D - I \end{bmatrix} V} \geq 0$.
Notice that  this is equivalent to the optimal value of \eqref{eq:convex-h-infinity} is greater than equal to $1$, since the optimum is always attained. 
Since \eqref{eq:convex-h-infinity} is equivalent to  \eqref{eq:original-h-infinity}, we can conclude the proof.
\end{proof}
\begin{theorem}[KYP lemma, non-strict inequality]
Suppose $(A,B)$ is controllable.
Then $\mu_{\infty} \leq 1$ if and only if there exists $P \in \mathbb{H}^n$ such that 
\begin{equation}
\begin{bmatrix} A^*PA-P & A^*PB\\B^*PA & B^*PB \end{bmatrix}+ \begin{bmatrix} C^*C & C^*D\\ D^*C & D^*D-I_m \end{bmatrix} \preceq 0 \label{eq:kypnp}.
\end{equation}
\label{thm:strict-kyp}
\end{theorem}
\begin{proof}
If $\mu_{\infty} < 1$, then from the Theorem \ref{thm:strict-kyp} the result is obvious.
Suppose the optimal value of \eqref{eq:original-h-infinity} is $1$. Then the optimal value of \eqref{eq:convex-h-infinity} is also $1$.
From the Corollary \ref{cor:solvable}, this is equivalent to the existence of the dual optimal solution $(\lambda, P) = (1, P^{\star})$, and this concludes the proof.
\end{proof}

\section{Extended $\mathcal{H}_{\infty}$ analysis}
In $\mathcal{H}_{\infty}$ analysis, a disturbance $\mathbf{w}$ is assumed to have an unit energy, $\|\mathbf{w}\|_2 = 1$.
Suppose more information about a disturbance is known beforehand.
Then $\mathcal{H}_{\infty}$ norm becomes conservative since the analysis does not exploit this additional information.
Therefore, it is natural to ask a question whether we can capture more general disturbance sets beyond $\|\mathbf{w}\|_2 = 1$, and formulate appropriate $\mathcal{H}_{\infty}$ optimization.
In this section, we propose extended $\mathcal{H}_{\infty}$ analysis with different constraints on the disturbance.
Some of these results are known from \cite{d1995h} in the form of the scaled small gain test, 
but we explicitly propose a well-defined optimization problem which contains a disturbance information which allows us to extract the worst-case disturbance,
and provide its exactness without additional effort.

Now consider the following disturbance set.
\BEAS
\mathcal{W} = \left\{\mathbf{w} \in l_2: f_i\left(\Lambda(\mathbf{w})\right) \preceq 0, i=1,\cdots,n_c\right\},
\EEAS
where each $f_i$ is a matrix valued linear function maps $\mathbb{H}$ to $\mathbb{H}$.
This set can be used to capture some interesting prior knowledge on the disturbance.
\begin{example}[$\mathcal{H}_{\infty}$ analysis]
In the $\mathcal{H}_{\infty}$ analysis, we require $\|\mathbf{w}\|_2 = 1$.
Using $f_1(W) = \trace{W}-1$, $f_2(W) = 1-\trace{W}$, then, 
$$\mathcal{W} = \left\{\mathbf{w} \in l_2: \trace{\Lambda(\mathbf{w})} = \|\mathbf{w}\|_2^2 = 1\right\},$$
which is desired.
\end{example}
\begin{example}[Square $\mathcal{H}_{\infty}$ analysis]
In \cite{d1995h}, the disturbance $\mathbf{w}$ satisfies $\|\mathbf{w}_i\|_2 \leq 1$, for $i = 1,\cdots,m$, where $\mathbf{w}_i$ is the $i$th component of $\mathbf{w}$.
Using $f_i(W) = W_{ii}-1$, for $i = 1,\cdots,m$, then,
$$\mathcal{W} = \left\{\mathbf{w} \in l_2: f_i\left(\Lambda(\mathbf{w})\right) = \|\mathbf{w}_i\|_2^2 \leq 1, i=1,\cdots,m\right\},$$
which is desired.
\label{ex:square}
\end{example}
\begin{example}[Grouped square $\mathcal{H}_{\infty}$ analysis]
Suppose the disturbance is in $l_2^4$, and  $\|\mathbf{w}_1\|_2^2 + \|\mathbf{w}_2\|_2^2 \leq 1$, $\|\mathbf{w}_3\|_2^2 +\|\mathbf{w}_4\|_2^2 \leq 1$.
Then using $f_1(W) = W_{11} + W_{22} - 1$, $f_2(W) = W_{33} + W_{44} - 1$, we can capture this disturbance.
\end{example}
\begin{example}[Principal component bound]
Suppose we know that the maximum eigenvalue of the autocovariance of the signal $\mathbf{w}$ is bounded by one.
This can be easily captured by $f_1(W) = W - I_m$. and
$\mathcal{W} = \left\{\mathbf{w} \in l_2: \Lambda(\mathbf{w}) \preceq I_m\right\}$.
See \cite{paganini1999convex} and references therein for the application of this disturbance modeling.
\end{example}

These examples show that our modeling framework can capture various information on the disturbance.
Now the next step is to find a computationally tractable method to analyze the worst-case performance as in the $\mathcal{H}_{\infty}$ case.
In other words, we would like to find a way to solve the following infinite dimensional optimization.
\begin{equation}
\begin{aligned}
& \underset{\mathbf{w},\mathbf{x}} {\text{maximize}}
& & \|C\mathbf{x} + D\mathbf{w}\|_2^2\\
&\text{subject to}
&& x_{k+1} = Ax_k + Bw_k, x_0 = 0\\
&&& \mathbf{w} \in \mathcal{W}.
\end{aligned}
\label{eq:original-general-h-infinity}
\end{equation}
Using the Gramian $V = \begin{bmatrix} X & R\\R^* & W\end{bmatrix} = \Lambda(\mathbf{Mw}, \mathbf{w})$ as in the $\mathcal{H}_{\infty}$ analysis, the above optimization is equivalent to

\begin{equation}
\begin{aligned}
& \underset{X,R,W} {\text{maximize}}
& & \trace{\begin{bmatrix}C^*C & C^*D\\D^*C&D^*D\end{bmatrix} \begin{bmatrix} X & R\\R^* & W\end{bmatrix}}\\
&\text{subject to}
&& V=\begin{bmatrix} X & R\\R^* & W\end{bmatrix} \in \mathcal{S}\\
&&& f_i(W)\preceq 0, \qquad i = 1,\cdots,n_c.
\end{aligned}
\label{eq:nonconvex-general-h-infinity}
\end{equation}

In $\mathcal{H}_{\infty}$ analysis, we replace $\mathcal{S}$ by $\mathcal{C}$ which is an SDP representable set.
For \eqref{eq:nonconvex-general-h-infinity}, we can also apply the same procedure.
The following proposition can be seen as a generalization of the Proposition \ref{prop:h-infinity}.
\begin{proposition}
Let $f_i:\mathbb{H} \rightarrow \mathbb{H}$, $i = 1\cdots, n_c$, be the linear function.
Define the set $\mathcal{F} = \{V=\begin{bmatrix} X & R\\R^* & W\end{bmatrix} \in \mathbb{H}:f_i(W) \preceq 0, i = 1,\cdots,n_c\}$.
Then $\cl({\mathcal{S} \cap \mathcal{F}}) = \mathcal{C} \cap \mathcal{F}$.
\label{prop:general-approx}
\end{proposition}
\begin{proof}
Since $\cl{\mathcal{S}} = \mathcal{C}$, and $\cl{\mathcal{F}} = \mathcal{F}$, we have $\cl{(\mathcal{S} \cap \mathcal{F})} \subset \mathcal{C} \cap \mathcal{F}$.
Now suppose $\begin{bmatrix} X & R\\R^* & W\end{bmatrix} \in  \mathcal{C} \cap \mathcal{F}$. 
From Lemma \ref{prop:approximation}, for any $\epsilon > 0$, there exists $\mathbf{w} \in l_2$ such that
$\Lambda(\mathbf{w}) = W$, and $\|V-\Lambda(\mathbf{Mw}, \mathbf{w})\|_F < \epsilon$.
Since $\Lambda(\mathbf{Mw}, \mathbf{w}) \in \mathcal{S} \cap \mathcal{F}$, we can conclude that $V \in \cl({\mathcal{S} \cap \mathcal{F}})$.
\end{proof}

Therefore, the optimization \eqref{eq:nonconvex-general-h-infinity} is equivalent to the following SDP.
\begin{equation}
\begin{aligned}
& \underset{X,R,W} {\text{maximize}}
& & \trace{\begin{bmatrix}C^*C & C^*D\\D^*C&D^*D\end{bmatrix} \begin{bmatrix} X & R\\R^* & W\end{bmatrix}}\\
&\text{subject to}
&&  X =  \begin{bmatrix} A & B \end{bmatrix}\begin{bmatrix} X & R\\R^* & W\end{bmatrix}\begin{bmatrix} A^*\\ B^* \end{bmatrix}\\
&&& \begin{bmatrix} X & R\\R^* & W\end{bmatrix} \succeq 0,\\
&&& f_i(W)\preceq 0, \qquad i = 1,\cdots,n_c.
\end{aligned}
\label{eq:convex-general-h-infinity}
\end{equation}

Notice that as in the $\mathcal{H}_{\infty}$ analysis case, once the optimal solution of \eqref{eq:convex-general-h-infinity} is obtained, we can explicitly construct the worst case disturbance $\mathbf{w}$ (approximately) achieves the maximum value, and this is not available in \cite{d1995h}.
 
Moreover, it is easy to derive the dual program using SDP duality and the KYP lemma like result can be obtained for a given disturbance set.
For example, applying the Theorem \ref{thm:alt4} to the Example \ref{ex:square}, the square $\mathcal{H}_{\infty}$ analysis, 
the optimal value of \eqref{eq:convex-general-h-infinity} is less than $1$ if and only if
there exists $P, Y$ such that
\BEAS
Y &\succeq& 0\\
\trace{Y} &<& 1\\
 \begin{bmatrix} A^*PA-P & A^*PB\\B^*PA & B^*PB \end{bmatrix}+ \begin{bmatrix} C^*C & C^*D\\ D^*C & D^*D-Y \end{bmatrix} &\prec& 0,
\EEAS
where $Y$ is a diagonal matrix. Notice that this result can also be found in \cite{d1995h}, but the proof in here is significantly simplified.

Some of these results may be obtained obtained through the $\mathcal{S}$-procedure, but proving the $\mathcal{S}$-procedure is lossless is not a trivial task \cite{polik2007survey}.
However, our approach shows that the $\mathcal{S}$-procedure can be viewed as a SDP relaxation of \eqref{eq:nonconvex-general-h-infinity}.
This idea may trace back to \cite{scherer2006lmi}, but the direct, primal formulation of the problem firstly appears here to the best of our knowledge.

\section{Robust stability analysis}
The robust stability analysis investigates the stability of the feedback interconnection between the nominal plant $\mathbf{G}$, a bounded operator from $l_2$ to  $l_2$, and the uncertain operator ${\Delta}$, a bounded operator from $l_2$ to $l_2$, which belongs to a set $\mathbf{\Delta}$.
The plant $\mathbf{G}$ is said to be robustly stable with respect to $\mathbf{\Delta}$ if 
\BEAS
\mathbf{I}-\Delta\mathbf{G},
\EEAS
is non-singular for all $\Delta \in \mathbf{\Delta}$.
See \cite{Dullerud:2010tc} and references therein.

In this section, we investigate the possible connection between the robust stability analysis and our key Lemmas, Lemma \ref{prop:approximation} and \ref{prop:closure}. 

To begin with, we assume that the uncertainty set $\mathbf{\Delta}$ admits the equivalent input-output characterization \cite{paganini1994behavioral, d1993uncertain}, that is, there exists a set $\mathcal{R}_a$ such that
\BEAS
\mathcal{R}_a := \{(\mathbf{z}, \mathbf{w}): \text{There exists $\Delta \in \mathbf{\Delta}$ such that $\mathbf{w} = \Delta \mathbf{z}$}\}.
\EEAS

Moreover, we assume that $\mathcal{R}_a$ can be completely characterized by a linear map, \ie, there exists  $f: \mathbb{H} \rightarrow \mathbb{H}$ such that $(\mathbf{z},\mathbf{w}) \in \mathcal{R}_a$ if and only if
\BEAS
f(\Lambda(\mathbf{z},\mathbf{w})) \succeq 0.
\EEAS
Let $\Lambda(\mathbf{z},\mathbf{w}) = \begin{bmatrix} Z & R\\R^* & W \end{bmatrix}$, then the following examples show that how $f$ can be used to describe $\mathcal{R}_a$.

\begin{example}[Full block complex LTV]
Suppose $\mathbf{\Delta} = \{\Delta : \|\Delta\| \leq 1\}$. 
Then $(\mathbf{z},\mathbf{w}) \in \mathcal{R}_a$ if and only if $\|\mathbf{w}\|_2 \leq \|\mathbf{z}\|_2$.
Therefore,
$f(\Lambda(\mathbf{z},\mathbf{w})) = \trace{Z} - \trace{W}$.
\label{ex:fullLTV}
\end{example}
\begin{example}[Two block complex LTV, \cite{shamma1994robust}]
Suppose $\mathbf{\Delta} = \{\Delta : \Delta = \begin{bmatrix} \Delta_1 & 0\\ 0& \Delta_2 \end{bmatrix}, \|\Delta_i\| \leq 1\}$. 
Then $(\mathbf{z},\mathbf{w}) \in \mathcal{R}_a$ if and only if $\|\mathbf{w}_1\|_2 \leq \|\mathbf{z}_1\|_2$, and $\|\mathbf{w}_2\|_2 \leq \|\mathbf{z}_2\|_2$, where $\mathbf{z}_i, \mathbf{w}_i$ are appropriately partitioned according to the size of $\Delta_i$.
In this case,
$f(\Lambda(\mathbf{z},\mathbf{w})) = \begin{bmatrix} \trace{Z_1} - \trace{W_1} &0\\ 0& \trace{Z_2} - \trace{W_2} \end{bmatrix}$.
\label{example:two-block}
\end{example}
\begin{example}[Scalar block complex LTV, \cite{paganini1996sets}]
Suppose $\mathbf{\Delta} = \{\Delta : \Delta = \delta \mathbf{I}, \|\delta\| \leq 1\}$.
Then $(\mathbf{z},\mathbf{w}) \in \mathcal{R}_a$ if and only if $\Lambda(\mathbf{w}) \preceq \Lambda(\mathbf{z})$.
Therefore,
$f(\Lambda(\mathbf{z},\mathbf{w})) = Z-W$.
\label{ex:scalarLTV}
\end{example}
\begin{example}[Integral quadratic constraints, \cite{megretski1997system}]
Suppose $(\mathbf{z},\mathbf{w}) \in \mathcal{R}_a$ if and only if
\BEAS
\sum_{k=0}^{\infty} \begin{bmatrix} z_k\\w_k\end{bmatrix}^*\Pi\begin{bmatrix} z_k\\w_k\end{bmatrix} \geq 0.
\EEAS
This is equivalent to
\BEAS
\trace{\Pi \begin{bmatrix} Z & R\\R^* & W \end{bmatrix}} \geq 0.
\EEAS
\label{ex:IQC}
\end{example}

Notice that using this set description of $\mathcal{R}_a$, we have
\BEAS
&&(\mathbf{I}-\Delta\mathbf{G})\mathbf{w} = 0\\
&{\Leftrightarrow}& \mathbf{w} = \Delta \mathbf{Gw}\\
&{\Leftrightarrow}& (\mathbf{Gw},\mathbf{w}) \in \mathcal{R}_a,
\EEAS
This means if there exists $\|\mathbf{w}\| = 1$ such that $(\mathbf{Gw},\mathbf{w}) \in \mathcal{R}_a$, then $\mathbf{G}$ is not robustly stable.
Therefore it is natural to consider the following set of values generated by $\mathcal{G}$:
\BEAS
\mathcal{G} := \{f(\Lambda(\mathbf{Gw}, \mathbf{w})): \|\mathbf{w}\| = 1\},
\EEAS
then we can easily see that $\mathcal{G} \cap \mathbb{H}_+ \neq \varnothing$ is the equivalent condition for the existence of $\|\mathbf{w}\| = 1$, such that$(\mathbf{Gw},\mathbf{w}) \in \mathcal{R}_a$.
In other words, $\mathcal{G} \cap \mathbb{H}_+ \neq \varnothing$ then the system $\mathbf{G}$ cannot be robustly stable.

Therefore we can change the robust stability question to a set relationship question, but characterizing $\mathcal{G}$ may not be trivial.
However, if $\mathbf{G}$ has a state-space form such that $\mathbf{x} = \mathbf{Mw}$, and $\mathbf{z} = C\mathbf{x} + D\mathbf{w}$,
then we have the following proposition which is a direct consequence of our main result.

\begin{proposition}
Suppose $\mathbf{G}(e^{j\theta}) = C(e^{j\theta}I -A)^{-1}B + D$. Then the closure of $\mathcal{G}$ is given by
\BEAS
\cl\mathcal{G} = \Big\{f\left(\begin{bmatrix} C & D\\ 0 & I\end{bmatrix} \begin{bmatrix} X & R\\R^* & W\end{bmatrix}\begin{bmatrix} C & D\\ 0 & I\end{bmatrix}^*\right) : \\
\trace{W} = 1, \begin{bmatrix} X & R\\R^* & W\end{bmatrix} \in \mathcal{C}\Big\}
\EEAS
\end{proposition}
\begin{proof}
Let $\mathbf{M} = (e^{j\theta}I -A)^{-1}B$.
Then, $\mathbf{z} = \mathbf{Gw} = C \mathbf{Mw} + D \mathbf{w}$.
Now
\begin{align*}
&\Lambda(\mathbf{Gw}, \mathbf{w}) = \Lambda(C\mathbf{Mw} + D\mathbf{w}, \mathbf{w})
= \Lambda\left(\begin{bmatrix} C & D\\ 0 & I \end{bmatrix} \begin{bmatrix} \mathbf{Mw}\\
\mathbf{w}
\end{bmatrix}\right)\\
&= \begin{bmatrix} C & D\\ 0 & I \end{bmatrix} \Lambda(\mathbf{Mw},\mathbf{w}) \begin{bmatrix} C & D\\ 0 & I \end{bmatrix}^*.
\end{align*}
This shows
\BEAS
\mathcal{G} = \Big\{f\left(\begin{bmatrix} C & D\\ 0 & I\end{bmatrix} \begin{bmatrix} X & R\\R^* & W\end{bmatrix}\begin{bmatrix} C & D\\ 0 & I\end{bmatrix}^*\right) : \\
\trace{W} = 1, \begin{bmatrix} X & R\\R^* & W\end{bmatrix} \in \mathcal{S}\Big\}
\EEAS
Recall that $\cl{\mathcal{S}} = \mathcal{C}$. Using the continuity of $f$, we can conclude the proof.
\end{proof}

Since $\mathcal{C}$ is an SDP representable cone, the above characterization is much easier to handle compared to $\mathcal{G}$.
In fact, using the theorem of alternatives, we can obtain the very interesting LMI characterization of $\cl\mathcal{G} \cap \mathbb{H}_+ \neq \varnothing$.

\begin{theorem}[ALT5a]
Exactly one of the following is true.
\begin{enumerate}[(i)]
\item $\exists X \in \mathbb{H}$ such that $\mathcal{A}(X) \succeq 0$, $\mathcal{A}(X) \neq 0$, and $\mathcal{B}(X) = 0$.
\item $\exists Z \in \mathbb{H}_{++}$, $W \in \mathbb{H}$, such that $\mathcal{A}^{*}(Z) + \mathcal{B}^*(W) = 0$.
\end{enumerate}
\label{thm:alt5a}
\end{theorem}

\begin{proposition}
Exactly one of the following is true.
\begin{enumerate}[(i)]
\item $\cl{\mathcal{G}} \cap \mathbb{H}_+ = \varnothing$.
\item There exists $P \in \mathbb{H}$, $Y \succ 0$ such that
\begin{equation}
\hspace{-5mm}
\begin{bmatrix}
A^*PA-P & A^*PB\\
B^*PA & B^*PB
\end{bmatrix} + 
\begin{bmatrix}
C & D\\
0 & I
\end{bmatrix}^*
f^{*}(Y)
\begin{bmatrix}
C & D\\
0 & I
\end{bmatrix}
\prec 0
\label{eq:scaled-small-gain}
\end{equation}
\end{enumerate}
\label{prop:multipler}
\end{proposition}
\begin{proof}
Let
\BEAS
\mathcal{A}(V) &=& \begin{bmatrix} V & \\ & f\left(\begin{bmatrix} C & D\\0 &I\end{bmatrix} V \begin{bmatrix} C & D\\0 & I\end{bmatrix}^*\right) \end{bmatrix}\\
\mathcal{B}(V) &=& \begin{bmatrix} A & B\end{bmatrix}V\begin{bmatrix} A & B\end{bmatrix}^* -  \begin{bmatrix} I & 0\end{bmatrix}V\begin{bmatrix} I & 0\end{bmatrix}^*.
\EEAS
Then $\cl{\mathcal{G}} \cap \mathbb{H}_+ =\varnothing$ if and only if there exists $V$ such that $\mathcal{A}(V) \succeq 0$, but $\mathcal{A}(V) \neq 0$, and $\mathcal{B}(V) = 0$. We can resale $V$, if necessary, to satisfy the trace condition.
Notice that the adjoint of the right bottom block of $\mathcal{A}(V)$ is given by
\BEAS
&&\ip{Y}{(\mathcal{A}(V))_{22}} = 
\trace{ Y  f\left(\begin{bmatrix} C & D\\0 & I\end{bmatrix} V \begin{bmatrix} C & D\\0 & I\end{bmatrix}^*\right)}\\
&=& \trace{f^{*}(Y) \begin{bmatrix} C & D\\0 & I\end{bmatrix} V \begin{bmatrix} C & D\\0 & I\end{bmatrix}^*}\\
&=& \trace{ \begin{bmatrix} C & D\\0 & I\end{bmatrix}^*f^{*}(Y)\begin{bmatrix} C & D\\0 & I\end{bmatrix}V}.
\EEAS
Using this fact with Theorem \ref{thm:alt5a}, we can conclude the proof.
\end{proof}

The above Proposition is indeed very interesting. It presents the equivalent LMI characterization of separating two sets.
Let us apply the above result to the previous examples of $\Delta$.

For Example \ref{ex:fullLTV}, $f(V) = \trace{\begin{bmatrix} I & 0\\0 & -I\end{bmatrix} V}$, and $f^*(y) = \begin{bmatrix} yI & 0 \\ 0 & -yI \end{bmatrix}$,
where the domain of $f^*$ is $\mathbb{H}^1$.
In this case, condition \eqref{eq:scaled-small-gain} becomes
\begin{align*}
&\begin{bmatrix}
A & B\\
0 & I
\end{bmatrix}^*
\begin{bmatrix}
P & 0\\
0 & -P
\end{bmatrix}
\begin{bmatrix}
A & B\\
0 & I
\end{bmatrix} \\
&+ 
\begin{bmatrix}
C & D\\
0 & I
\end{bmatrix}^*
\begin{bmatrix}
yI & 0\\
0 & -yI
\end{bmatrix}
\begin{bmatrix}
C & D\\
0 & I
\end{bmatrix} \prec 0,
\end{align*}
and since $y > 0$, by multiplying $1/y$ to both sides, we recover the KYP lemma for $\|\mathbf{G}\|_{\infty} < 1$.

For Example \ref{ex:scalarLTV}, $f(V) = \begin{bmatrix}I & 0\end{bmatrix}V\begin{bmatrix}I \\ 0\end{bmatrix} - \begin{bmatrix}0 & I\end{bmatrix}V\begin{bmatrix}0\\ I\end{bmatrix}$, and $f^*(Y) = \begin{bmatrix} Y & 0 \\ 0 & -Y \end{bmatrix}$, where the domain of $f^*$ is $\mathbb{H}^m$.
In this case, condition \eqref{eq:scaled-small-gain} becomes
\begin{align*}
&\begin{bmatrix}
A & B\\
0 & I
\end{bmatrix}^*
\begin{bmatrix}
P & 0\\
0 & -P
\end{bmatrix}
\begin{bmatrix}
A & B\\
0 & I
\end{bmatrix} \\
&+ 
\begin{bmatrix}
C & D\\
0 & I
\end{bmatrix}^*
\begin{bmatrix}
Y & 0\\
0 & -Y
\end{bmatrix}
\begin{bmatrix}
C & D\\
0 & I
\end{bmatrix} \prec 0,
\end{align*}

Since $Y \succ 0$, by left and right multiplying $Y^{-1/2}$ to the above expression, we can recover $\|Y^{-1/2}\mathbf{G}Y^{1/2}\|_{\infty} < 1$, which is a scaled small gain test.

Finally, for Example \ref{ex:IQC}, we have $f(V) = \trace{\Pi V}$, and $f^*(y) = y\Pi$, where the domain is $\mathbb{H}^1$.
In this case, condition \eqref{eq:scaled-small-gain} becomes
\BEAS
&&\begin{bmatrix}
A^*PA-P & A^*PB\\
B^*PA & B^*PB
\end{bmatrix} 
+
\begin{bmatrix}
C & D\\
0 & I
\end{bmatrix}^*
\Pi
\begin{bmatrix}
C & D\\
0 & I
\end{bmatrix} \prec 0.
\EEAS

For block diagonal structure, such as Example \ref{example:two-block} we can also apply our approach to recover the block diagonal small gain test as in \cite{shamma1994robust}.

From robust control theory, we know that the scaled small gain test provides a sufficient and necessary test for robust stability in certain cases \cite{Dullerud:2010tc}. 
In fact, all the examples we provide fall in to those classes.
Notice that we have shown that the scaled small gain test is a necessary and sufficient condition for $\cl{\mathcal{G}} \cap \mathbb{H}_+ = \varnothing$.
Therefore, we obtain the following chain of equivalent statements when the scaled small gain test becomes the equivalent condition for the robust stability of $\mathbf{G}$.
\BEAS
&&\cl{\mathcal{G}} \cap \mathbb{H}_+ = \varnothing \overset{(a)}{\Leftrightarrow} \|\Theta^{-1} \mathbf{G} \Theta\|_{\infty} < 1 \\
&\overset{(b)}\Leftrightarrow &\text{Robust stability of $\mathbf{G}$}
\EEAS
However, proving (b) from the scaled small gain test is not a trivial task whereas  (a) is from a standard machinery.
We strongly believe that there exists a direct proof between $(\cl{\mathcal{G}} \cap \mathbb{H}_+ = \varnothing) {\Leftrightarrow} (\text{Robust stability of $\mathbf{G}$})$ in the style of \cite{safonov1980stability} without relying on the complicated argument, and this is currently under investigation.

Another important observation is that if $\cl{\mathcal{G}} \cap \mathbb{H}_+ \neq \varnothing$, then we can use SDP to find such a $V$.
Then using Lemma \ref{prop:approximation}, we can find a pair $(\mathbf{w}, \mathbf{z})$ that approximately satisfies
\BEAS
\mathbf{z} = \mathbf{Gw}\\
\mathbf{w} = \Delta\mathbf{z}.
\EEAS
Notice that this pair $(\mathbf{w}, \mathbf{z})$ disproves the robust stability.

In conclusion, we have shown how our main result can be used to derive the scaled small gain test without using the commutant, $\Theta$.
Moreover, we show that the exact implication of the scaled small gain test, $\cl{\mathcal{G}} \cap \mathbb{H}_+ = \varnothing$, using our results.
This suggests that there may exist deeper connection between robust stability and the set relationship $\cl{\mathcal{G}} \cap \mathbb{H}_+ = \varnothing$, and we are currently investigating their exact relationship.

\section{Conclusion}
In this paper, we propose an alternative, reverse direction of theoretical development for robust control theory.
Based on our Lemma \ref{prop:approximation}, an SDP representation of a set of gramians, we show that the robust analysis and stability question can be directly formulated as a primal optimization.
Moreover, we show that the well-known results in robust control theory can be obtained via SDP duality, an arguably simple machinery to prove many interesting results,
and this shows that our approach is a primal formulation of robustness analysis.
Therefore, we believe that our paper provides an alternative "primal-dual" picture in robust control theory, and this new development not only opens up the new  research direction, but also enhances pedagogy.

\section{Appendix: Proof of the main result}
This section we present the main proof of our proposition \ref{prop:approximation}, with technical lemmas that need to prove our result.

\subsection{Preliminaries from linear algebra}
The following results from the linear algebra are used to prove the main result.
\begin{proposition}
For $x, y \in \mathbb{C}^{n}$, $\|xy^*\|_F = \|x\|_2\|y\|_2$.
\end{proposition}
\begin{proof}
$\|xy^*\|^2_F = \trace{(xy^*)^*(yx^*)} = \|x\|_2^2\|y\|_2^2$.
\end{proof}
\begin{theorem}[Gelfand, 1941]
For any matrix norm $\|\cdot\|$, 
\BEAS
\lim_{n\rightarrow \infty} \|A^n\|^{1/n} = \rho(A).
\EEAS
\label{thm:gelfand}
\end{theorem}
\begin{proposition}
Suppose $\rho(A) < 1$. Then, for any matrix norm $\|\cdot\|$, $\|A^k\| \in l_1$.
\label{prop:A-l1}
\end{proposition}
\begin{proof}
Let $\epsilon = (1-\rho(A)) / 2 > 0$. Then from the Theorem \ref{thm:gelfand}, there exists $N \in \mathbb{N}$ such that
\BEAS
\|A^k\| < (\rho(A)+\epsilon)^k,
\EEAS
for all $k \geq N$. 
\BEAS
\sum_{k=0}^{\infty} \|A^k\| &=& \sum_{k=0}^N \|A^k\| + \sum_{k=N+1}^{\infty} \|A^k\|\\
	&<&\sum_{k=0}^N \|A^k\| + \sum_{k=N+1}^{\infty} (\rho(A)+\epsilon)^k.
\EEAS
Since $\rho(A)+\epsilon < 1$, the second term is finite. Therefore, $\|A^k\| \in l_1$.
\end{proof}

\begin{lemma}[Rantzer, 1996]
Let $F,G$ complex matrices with same dimension.
Then $FF^* = GG^*$ if and only if there exists a unitary matrix $U \in \mathbb{C}^{k \times k}$ such that $F = GU$.
\label{lemma:ra}
\end{lemma}
\begin{proof}
See \cite{Rantzer:2011wn}.
\end{proof}

The following result about a convex cone $\mathcal{C}$ is a direct consequence of the above lemma, and the results states that extreme points of $\mathcal{C}$ are rank one matrices.
\begin{proposition}[Rank one decomposition]
For all $V \in \mathcal{C}$, there exists a set of matrices $V_1,\cdots,V_{n+m} \in \mathcal{C}$ such that
$V = \sum_{k=1}^{n+m} V_k$, and $\rank{V_k} \leq 1$ for all $k = 1,\cdots,n+m$.
\label{lemma:rank1-1}
\end{proposition}
\begin{proof}
See \cite{you2014h}.
\end{proof}

\subsection{Technical lemmas for the main result}
In this section, we derive technical results to prove the main result of this paper. 
The basic idea behind the following results is to bound the error terms arise in the proof of our main result.

\begin{proposition}
Suppose $\mathbf{w}$ has finite number of non-zero entries. 
For any $\varepsilon > 0$, there exists $N \in \mathbb{N}$, such that
\BEAS
\left\|\Lambda(\mathbf{Mw}, \mathbf{w}) - \Lambda_n(\mathbf{Mw}, \mathbf{w})\right\|_F < \varepsilon,
\EEAS
for all $n \geq N$.
\label{prop:finiteapproximiation}
\end{proposition}
\begin{proof}
Let $\mathbf{x} = \mathbf{Mw}$, and consider $T \in \mathbb{N}$ such that $w_k = 0$, for all $k \geq T$.
For $N \geq T$, we have
\BEAS
&&\left\|\Lambda(\mathbf{Mw}, \mathbf{w}) - \Lambda_{N}(\mathbf{Mw}, \mathbf{w})\right\|_F\\
&=& \left\|\sum_{k=N}^{\infty} A^{k}x_T(A^{k}x_T)^*\right\|_F
\leq \sum_{k=N}^{\infty} \|A^{k}x_T\|_2^2\\
&=& \|A^{N}x_T\|_2^2\sum_{k=0}^{\infty} \|A^k\|_2^2
\EEAS
Since $\|A^k\|_2 \in l_1 \subset l_2$, the infinite sum is finite. In addition, from the Theorem \ref{thm:gelfand}, $\|A^{N}x_T\|_F \rightarrow 0$ as $N \rightarrow \infty$.
Therefore by choosing $N$ sufficiently large, we can conclude the proof.
\end{proof}

\begin{proposition}
Let $\mathbf{x}, \mathbf{w} \in l_2$, and $\mathbf{y} \in l_2$ such that $y_k = A^k y_0$ for all $k$.
Then there exists a constant $C$ such that
\BEAS
&&\left\|\Lambda(\mathbf{x} + \mathbf{y},\mathbf{w}) - \Lambda(\mathbf{x},\mathbf{w})\right\| \\
&\leq& C\max\{(|\mathbf{x}\|_{\infty} + \|\mathbf{w}\|_{\infty})\|y_0\|_{2}, \|y_0\|_{2}^2\}.
\EEAS
\label{prop:A-perturbation}
\end{proposition}
\begin{proof}
Notice that,
\BEAS
&&\left\|\Lambda(\mathbf{x} + \mathbf{y},\mathbf{w}) - \Lambda(\mathbf{x},\mathbf{w})\right\|_F\\
&=& \left\|\sum_{k=0}^{\infty}\begin{bmatrix}{x}_k+y_k\\w_k\end{bmatrix}\begin{bmatrix}{x}_k+y_k\\w_k\end{bmatrix}^* - \begin{bmatrix}{x}_k\\w_k\end{bmatrix}\begin{bmatrix}{x}_k\\w_k\end{bmatrix}^*\right\|_F\\
&=& \left\|\sum_{k=0}^{\infty}\begin{bmatrix}y_k\\0\end{bmatrix}\begin{bmatrix}x_k\\w_k\end{bmatrix}^* + \begin{bmatrix}x_k\\w_k\end{bmatrix}\begin{bmatrix}y_k\\0\end{bmatrix}^*+ \begin{bmatrix}y_k\\0\end{bmatrix}\begin{bmatrix}y_k\\0\end{bmatrix}^*\right\|_F\\
&\leq& \sum_{k=0}^{\infty}2\left\|\begin{bmatrix}y_k\\0\end{bmatrix}\begin{bmatrix}x_k\\w_k\end{bmatrix}^*\right\|_F + \left\|\begin{bmatrix}y_k\\0\end{bmatrix}\begin{bmatrix}y_k\\0\end{bmatrix}^*\right\|_F\\
&=& \sum_{k=0}^{\infty}2\|y_k\|_2 \sqrt{\|x_k\|_2^2 + \|w_k\|_2^2} + \|y_k\|_2^2.
\EEAS
Since $\mathbf{x}, \mathbf{w} \in l_2 \subset l_2$, we have
\BEAS
\sqrt{\|x_k\|_2^2 + \|w_k\|_2^2} \leq |\mathbf{x}\|_{\infty} + \|\mathbf{w}\|_{\infty},
\EEAS
for all $k$.
Moreover, since $y_k = A^ky_0$, we have
\BEAS
&&\left\|\Lambda(\mathbf{x} + \mathbf{y},\mathbf{w}) - \Lambda(\mathbf{x},\mathbf{w})\right\|_F\\
&\leq& \sum_{k=0}^{\infty}2(|\mathbf{x}\|_{\infty} + \|\mathbf{w}\|_{\infty})\|A^ky_0\|_2 + \|A^ky_0\|_2^2\\
&\leq& \sum_{k=0}^{\infty} 2(|\mathbf{x}\|_{\infty} + \|\mathbf{w}\|_{\infty}) \|A^k\|_2\|y_0\|_2  + \|A^k\|_2^2\|y_0\|_2^2\\
&\leq&  C\max\{(|\mathbf{x}\|_{\infty} + \|\mathbf{w}\|_{\infty}) \|y_0\|_{2}, \|y_0\|_{2}^2\}
\EEAS
where $C = \left(\sum_{k=0}^{\infty}2\|A^k\|_2 + \|A^k\|_2^2\right)$. 
Since $\|A^k\|_2 \in l_1 \subset l_2$, $C < \infty$, and this concludes the proof.
\end{proof}

\subsection{Main result}
Now we are ready to prove main results of this paper.
The main idea of the proof is as follows.

Any rank one matrix in $\mathcal{C}$ can be generated by a sinusoid $\mathbf{w}$, but a sinusoid is not in $l_2$.
Therefore we find a signal in $l_2$ which approximates this sinusoid.
This is not surprising, because in $\mathcal{H}_{\infty}$ analysis \cite{doyle1992feedback}, the so called worst case signal is sinusoid which is not in $l_2$, so one has to approximate this sinusoid using $l_2$ and the supremum is not achieved.
More fundamental reason for this is due to non-compactness of an unit sphere in $l_2$, but we will not elaborate this point.

For a full rank matrix in $\mathcal{C}$, we firstly decompose this matrix to rank one matrices using the Lemma \ref{lemma:rank1-1}, 
then approximate each rank one matrices by a signal in $l_2$.
Finally, we pad them together to approximate a full rank matrix as in \cite{shamma1994robust}.

\begin{proposition}
Suppose $V \in \mathcal{C}$, and $\rank{V} \leq 1$.
Then for all $\varepsilon > 0$, there exists $\mathbf{w}$ with a finite number of non-zero entries such that
\begin{align}
&\|\Lambda(\mathbf{Mw},\mathbf{w}) - V\|_F < \varepsilon \label{eq:small-error}\\
&\Lambda(\mathbf{w}) = \begin{bmatrix} 0_{m,n} & I_{m}\end{bmatrix} V \begin{bmatrix} 0_{m,n} & I_{m}\end{bmatrix}^*  \label{eq:exactw}
\end{align}
\label{prop:rank1case}
\end{proposition}
\begin{proof}
Suppose $\rank{V} = 0$. Then $V = 0 \in \mathcal{C}$, and $\mathbf{w} = 0$ satisfies \eqref{eq:small-error} and \eqref{eq:exactw}.

Now suppose $\rank{V} = 1$. We will construct $\mathbf{w}$ which satisfies \eqref{eq:small-error} and  \eqref{eq:exactw}.
Since $\rank{V} = 1$, there exists $x_s \in \mathbb{C}^n, w_s \in \mathbb{C}^m$ such that $V = \begin{bmatrix} x_s \\ w_s \end{bmatrix} \begin{bmatrix} x_s \\ w_s \end{bmatrix}^*$.
Moreover, by defining $f = \begin{bmatrix} A & B\end{bmatrix}\begin{bmatrix} x_s \\ w_s \end{bmatrix}$, $g = \begin{bmatrix} I_n & 0_{n,m}\end{bmatrix}\begin{bmatrix} x_s \\ w_s \end{bmatrix}$, we can easily see that $ff^* = gg^*$.
Therefore from the Lemma \ref{lemma:ra}, there exists $\theta$ such that $e^{j\theta} x_s = Ax_s + Bw_s$.

Now, for a given $N \in \mathbb{N}$, define $\mathbf{w}$
\BEAS
w_k &=& \begin{cases}
	\frac{1}{\sqrt{N}} e^{j\theta k} w_{s} & \text{if $0 \leq k < N$}\\
	0 & \text{if $N \leq k$}.
	\end{cases}
\EEAS
It is easy to see that $\Lambda(\mathbf{w}) = \sum_{k=0}^{\infty} w_kw_k^* = \sum_{k=0}^{N-1} w_kw_k^* = w_sw_s^*$.
and therefore $\mathbf{w}$ satisfies \eqref{eq:exactw}.
In order to obtain $\mathbf{x} = \mathbf{Mw}$, let us define the following signal $\mathbf{s}, \mathbf{t}$
\BEAS
s_k &=&  \begin{cases}
	\frac{1}{\sqrt{N}} e^{j\theta_k} x_{s} & \text{if $0 \leq k < N$}\\
	\frac{1}{\sqrt{N}} A^{k-N}e^{j\theta N}x_s & \text{if $N \leq k$},
	\end{cases}\\
t_k &=& -\frac{1}{\sqrt{N}} A^kx_{s}
\EEAS
then $\mathbf{x} = \mathbf{s}+\mathbf{t}$.
Notice that
\BEAS
\Lambda_N(\mathbf{s}, \mathbf{w})
= \sum_{k=0}^{N-1}\begin{bmatrix} s_k\\w_k\end{bmatrix}\begin{bmatrix} s_k\\w_k\end{bmatrix}^* 
= \frac{1}{N}\sum_{k=0}^{N-1}\begin{bmatrix}x_s\\w_s\end{bmatrix}\begin{bmatrix}x_s\\w_s\end{bmatrix}^*
= V
\EEAS
and this shows
\BEAS
&&\|\Lambda(\mathbf{s},\mathbf{w}) -V\|_F
= \left\|\sum_{k=N}^{\infty} \begin{bmatrix}s_k\\w_k\end{bmatrix}\begin{bmatrix}s_k\\w_k\end{bmatrix}^*\right\|_F\\
&\leq& \frac{1}{N} \sum_{k=0}^{\infty} \left\|\begin{bmatrix}A^k x_s\\0\end{bmatrix}\begin{bmatrix}A^k x_s\\0\end{bmatrix}^*\right\|_F
\leq \frac{1}{N} \sum_{k=0}^{\infty} \|A^kx_s\|_2^2\\
&\leq& \underbrace{\sum_{k=0}^{\infty} \|A^k\|_2^2}_{C_1} \frac{\|x_s\|_2^2}{N}.
\EEAS
Notice that $C_1 < \infty$ because of the Proposition \ref{prop:A-l1}.

Now using the triangle inequality and the Proposition \ref{prop:A-perturbation}, we obtain
\BEAS
&&\|\Lambda(\mathbf{x},\mathbf{w}) - V\|_F 
= \|\Lambda(\mathbf{s}+\mathbf{t},\mathbf{w}) - V\|_F\\
&\leq& \|\Lambda(\mathbf{s},\mathbf{w}) - V\|_F + \|\Lambda(\mathbf{s}+\mathbf{t},\mathbf{w}) - \Lambda(\mathbf{s},\mathbf{w})\|_F\\
&\leq& C_1\frac{\|x_s\|_2^2}{N} + C\max\{(|\mathbf{s}\|_{\infty} + \|\mathbf{w}\|_{\infty})\|t_0\|_{2}, \|t_0\|_{2}^2\}
\EEAS
Since $\|\mathbf{s}\|_{\infty} = \frac{1}{\sqrt{N}} \|x_s\|$, $\|\mathbf{w}\|_{\infty} =\frac{1}{\sqrt{N}} \|w_s\|$, $\|t_0\|_2 = \frac{1}{\sqrt{N}} \|x_s\|$, we have
$
\max\{(|\mathbf{s}\|_{\infty} + \|\mathbf{w}\|_{\infty})\|t_0\|_{2}, \|t_0\|_{2}^2\} 
= \frac{1}{N} \|x_s\|_2(\|x_s\|_2 + \|w_s\|_2)$.
By combining all these bounds, we can conclude that there exists a constant $C_3$ only depends on $A, x_s, w_s$ such that
\BEAS
\|\Lambda(\mathbf{x},\mathbf{w}) - V\|_F  \leq \frac{C_3}{N},
\EEAS
therefore by taking sufficiently large $N$, we can make $\mathbf{w}$ satisfy \eqref{eq:small-error}, and clearly $\mathbf{w}$ has $N$ number of non-zero entries.
\end{proof}

Now combining all these results, we are ready to prove our main result, the Proposition \ref{prop:approximation}.
\begin{proof}[Proof of the Proposition \ref{prop:approximation}]
From the Lemma \ref{lemma:rank1-1}, we can decompose $V = \sum_{i=1}^{n+m} V_i$ where $V_i \in \mathcal{C}$, and $\rank{V_i} \leq 1$.
Let us rearrange these terms, so that $V = \sum_{i=1}^r V_i$ where $\rank{V_i} = 1$.
We now use an induction on $r$.
Suppose $r \leq 1$, then from the Proposition \ref{prop:rank1case}, the proof is done.

Now assume the induction hypothesis holds, that is for $\sum_{i=1}^{r-1}V_i \in \mathcal{C}$, there exists ${\tilde{\mathbf{w}}}$ with a finite number of non-zero entries 
such that
\BEAS
&&\left\|\Lambda(\mathbf{M\tilde{w}},\mathbf{\tilde{w}})- \sum_{i=1}^{r-1}V_i \right\|_F < \frac{1}{4}\varepsilon,\\
&&\Lambda(\mathbf{\tilde{w}}) = \sum_{i=1}^{r-1} \begin{bmatrix} 0_{m,n} & I_{m}\end{bmatrix} V_i \begin{bmatrix} 0_{m,n} & I_{m}\end{bmatrix}^* 
\EEAS
Similarly, for $V_r$, there exists $\hat{\mathbf{w}}$ with a finite number of non-zero entries such that
\BEAS
&&\left\|\Lambda(\mathbf{M\hat{w}},\mathbf{\hat{w}})- V_r\right\|_F < \frac{1}{4}\varepsilon,\\
&&\Lambda(\mathbf{\hat{w}}) =\begin{bmatrix} 0_{m,n} & I_{m}\end{bmatrix} V_r \begin{bmatrix} 0_{m,n} & I_{m}\end{bmatrix}^* 
\EEAS
satisfies  with $\varepsilon/4$ and 2.
In addition, since ${V_r} \in \mathcal{C}$ and $\rank{V_r} = 1$, from the Proposition \ref{prop:rank1case}, there exists finite length $\hat{\mathbf{w}}$ satisfies 1 with $\varepsilon/4$ and 2.

Let $T \in \mathbb{N}$ such that $\tilde{w}_k = 0$ for all $k \geq T$.
From the Proposition \ref{prop:finiteapproximiation}, we can find $N_1$ such that
\BEAS
\left\|\Lambda(\mathbf{M\tilde{w}},\mathbf{\tilde{w}})-\Lambda_{n+T}(\mathbf{M\tilde{w}},\mathbf{\tilde{w}})\right\| < \frac{1}{4}\epsilon,
\EEAS
for all $n \geq N_1$.

Consider the following signal $\mathbf{w}$ 
\BEAS
w_k = \begin{cases}
\tilde{w}_k & \text{if $0 \leq k < N+T$}\\
\hat{w}_{k-N-T} & \text{if $N+T \leq k$},
\end{cases}
\EEAS
where $N \geq N_1$.
Clearly, $\mathbf{w}$ has a finite number of non-zero entries, and $\Lambda(\mathbf{w}) = \Lambda(\mathbf{\tilde{w}}) +  \Lambda(\mathbf{\hat{w}})$, which shows 
\BEAS
\Lambda(\mathbf{w}) =\begin{bmatrix} 0_{m,n} & I_{m}\end{bmatrix} V \begin{bmatrix} 0_{m,n} & I_{m}\end{bmatrix}^*.
\EEAS

Now, let $\mathbf{\tilde{x}} = \mathbf{M\tilde{w}}$, $\mathbf{\hat{x}} = \mathbf{M\hat{w}}$, and $\mathbf{x} = \mathbf{Mw}$.
Then,
\BEAS
x_k  = \begin{cases}
\tilde{x}_k & \text{if $0 \leq k < T$}\\
A^{k-T} \tilde{x}_T & \text{if $T \leq k < N+T$}\\
\hat{x}_{k-N-T} + A^{k-T}\tilde{x}_T & \text{if $N+T \leq k$}.
\end{cases}
\EEAS

Notice that
\BEAS
\Lambda(\mathbf{x},\mathbf{w}) &=&
\Lambda_{N+T}(\mathbf{\tilde{x}},\mathbf{\tilde{w}}) \\
&&+ \sum_{k=0}^{\infty}
\begin{bmatrix}
\hat{x}_{k} + A^{k}A^{N}\tilde{x}_T\\
\hat{w}_{k}
\end{bmatrix}
\begin{bmatrix}
\hat{x}_{k} + A^{k}A^{N}\tilde{x}_T\\
\hat{w}_{k}
\end{bmatrix}^*\\
&=&
\Lambda_{N+T}(\mathbf{\tilde{x}}, \mathbf{\tilde{w}})
+ \Lambda(\mathbf{\hat{x}} + \mathbf{y}, \mathbf{\hat{w}}),
\EEAS
where $y_k = A^{k}A^{N}\tilde{x}_T$.

Therefore,
\BEAS
&&\|\Lambda(\mathbf{x}, \mathbf{w}) - V\|_F\\
&\leq& \left\|\Lambda_{N+T}(\mathbf{\tilde{x}}, \mathbf{\tilde{w}})-\sum_{i=1}^{r-1} V_i \right\|_F + \left\|\Lambda(\mathbf{\hat{x}} + \mathbf{y}, \mathbf{\hat{w}}) - V_r \right\|_F\\
&\leq& \left\|\Lambda_{N+T}(\mathbf{\tilde{x}}, \mathbf{\tilde{w}})-\Lambda(\mathbf{\tilde{x}}, \mathbf{\tilde{w}})\right\|_F 
+ \left\|\Lambda(\mathbf{\tilde{x}}, \mathbf{\tilde{w}})-\sum_{i=1}^{r-1} V_i \right\|_F\\
&&
+ \left\|\Lambda(\mathbf{\hat{x}} + \mathbf{y}, \mathbf{\hat{w}}) - \Lambda(\mathbf{\hat{x}}, \mathbf{\hat{w}}) \right\|_F
+ \left\|\Lambda(\mathbf{\hat{x}}, \mathbf{\hat{w}}) - V_r \right\|_F\\
&\leq& \frac{3}{4}\epsilon + C\max\{(|\mathbf{\hat{x}}\|_{\infty} + \|\mathbf{\hat{w}}\|_{\infty})\|y_0\|_{2}, \|y_0\|_{2}^2\}.
\EEAS
Notice that $\mathbf{\hat{x}}, \mathbf{\hat{w}}$ are given signals in $l_2$, and does not depend on our choice $N$.
However, since $\|y_0\|_2 = \|A^N\|\|\tilde{x}\|_T$ can be made arbitrarily small by taking $N \rightarrow \infty$,
and this concludes the proof.
\end{proof}


\bibliographystyle{IEEEtran}
\bibliography{IEEEabrv,gkyp}

\end{document}